%% file: main_v2.tex
\documentclass{article}
\usepackage{spconf,amsmath,graphicx,cite}
\usepackage{amsfonts,amssymb}
\usepackage{amsthm}
\usepackage{bm}
\usepackage{float}
\usepackage{tikz,pgfplots}
\usepackage{bm}
\usepackage{cases}
\usepackage{algorithm}
\usepackage{algorithmic}
\usepackage{multirow}
\usepackage{booktabs}
\usepackage{cases}
\usepackage{enumitem}
\usepackage{subcaption}
\usepgfplotslibrary{groupplots}

\usepackage[colorlinks=true,
linkcolor=blue,
urlcolor=blue,
citecolor=blue]{hyperref}  

\usepackage{cleveref}


\usepackage{shortcuts_OPT}

\newtheorem{Theorem}{Theorem}
\newtheorem{Def}{Definition}

\newtheorem{Corollary}{Corollary}

\newtheorem{assumption}{H\!\!}
\newtheorem{Remark}{Remark}
\crefname{Def}{Definition}{Definition} 

\usepackage{pgfplots}
\usetikzlibrary{arrows,shapes,calc,tikzmark,backgrounds,matrix,decorations.markings}
\usepgfplotslibrary{fillbetween}

\pgfplotsset{compat=1.3}

\usepackage{relsize}
\tikzset{fontscale/.style = {font=\relsize{#1}}
    }

\definecolor{lavander}{cmyk}{0,0.48,0,0}
\definecolor{violet}{cmyk}{0.79,0.88,0,0}
\definecolor{burntorange}{cmyk}{0,0.52,1,0}

\definecolor{asuorange}{rgb}{1,0.699,0.0625}
\definecolor{asured}{rgb}{0.598,0,0.199}
\definecolor{asuborder}{rgb}{0.953,0.484,0}
\definecolor{asugrey}{rgb}{0.309,0.332,0.340}
\definecolor{asublue}{rgb}{0,0.555,0.836}
\definecolor{asugold}{rgb}{1,0.777,0.008}

\DeclareUnicodeCharacter{2212}{−}
\usepgfplotslibrary{groupplots,dateplot}
\usetikzlibrary{patterns,shapes.arrows}
\pgfplotsset{compat=newest}

    \makeatletter
    \def\multilimits@{\bgroup
  \Let@
  \restore@math@cr
  \default@tag
 \baselineskip\fontdimen10 \scriptfont\tw@
 \advance\baselineskip\fontdimen12 \scriptfont\tw@
 \lineskip\thr@@\fontdimen8 \scriptfont\thr@@
 \lineskiplimit\lineskip
 \vbox\bgroup\ialign\bgroup\hfil$\m@th\scriptstyle{##}$\hfil\crcr}
    \def\Sb{_\multilimits@}
    \def\endSb{\crcr\egroup\egroup\egroup}
\makeatother

\makeatletter
\DeclareRobustCommand*\cal{\@fontswitch\relax\mathcal}
\makeatother

\begin{document}
\ninept

\title{On the Stability of Low Pass Graph Filter with A Large Number of Edge Rewires}
\name{Hoang-Son Nguyen, Yiran He, Hoi-To Wai\thanks{This work is supported by CUHK Direct Grant \#4055115. Emails: \texttt{sonngh.00@link.cuhk.edu.hk}, \texttt{\{yrhe,htwai\}@se.cuhk.edu.hk}}}
\address{The Chinese University of Hong Kong, Shatin, Hong Kong SAR of China.}

\maketitle
\begin{abstract}
Recently, the stability of graph filters has been studied as one of the key theoretical properties driving the highly successful graph convolutional neural networks (GCNs). The stability of a graph filter characterizes the effect of topology perturbation on the output of a graph filter, a fundamental building block for GCNs. Many existing results have focused on the regime of small perturbation with a small number of edge rewires. However, the number of edge rewires can be large in many applications. To study the latter case, this work departs from the previous analysis and proves a bound on the stability of graph filter relying on the filter's frequency response. Assuming the graph filter is low pass, we show that the stability of the filter depends on perturbation to the \emph{community structure}. As an application, we show that for stochastic block model graphs, the graph filter distance converges to zero when the number of nodes approaches infinity. Numerical simulations validate our findings.
\end{abstract}

\section{Introduction}\vspace{-.1cm}
Nowadays, developing models for graph structured data is popular across disciplines \cite{bronstein2017geometric, dong2020graph} with the increased availability of data in domains such as social networks and biological networks. One of the main thrusts is to exploit the irregular relationships between data samples through modeling the latter using graphs that encode arbitrary pairwise relationships. 
Among others, graph convolutional networks (GCNs) \cite{kipf2016semi,defferrard2016convolutional,wu2020comprehensive} are  shown to be efficient in learning graph structured data. 

Despite the highly successful applications, there has been few theoretical studies on \emph{``why GCNs are powerful models for graph structured data?''}. As GCNs are built on graph filters which replace the classical convolution operation in convolutional neural networks, one possible approach is to evaluate the efficacy of GCNs via analyzing the \emph{stability of graph filter} \cite{gama}. Formally, we consider an original graph and its perturbation such that we measure the maximum change in the graph filter output in $\ell_2$ norm when subjected to a unit norm excitation signal. The latter is also known as the graph filter distance. The concept of stability is related to transferability which studies the generalization property of a trained GCN when applied on datasets with different distribution or embedded graph topology.

Among the existing attempts, \cite{gama} provided a comprehensive study on the stability of graph filters utilizing an integral Lipschitz property, while \cite{levie2019transferability, kenlay_icassp20, kenlay_icassp21, kenlay_icml21} have derived bounds surrounding the polynomial graph filter models with normalized Laplacian matrices as graph shift operator (GSO); also see \cite{zou2020graph, gama2019stability, gao2021stability}. As a common trail, many of the above works yielded an upper bound on the {graph filter distance} as ${\cal O}( \| {\bf S} - \hat{\bf S} \|_2 )$ which is proportional to the number of edge rewires, where ${\bf S}, \hat{\bf S}$ are respectively the original and perturbed graph Laplacian/adjacency matrices. As such, they reach the conclusion that the stability of graph filter holds in scenarios when only a few edges are rewired in the event of graph perturbation. When the graph perturbation involves a large number of edge rewiring, the above bounds may become uninformative. This raises an interesting question: \emph{what are the conditions for graph filter to be stable when subjected to a large number of edge rewiring}?

This work provides an affirmative answer to the above question. Our intuition is that as long as the community structure of a graph is stable upon graph perturbation, e.g., only the edges within the same community are rewired, then the graph filter shall be stable regardless of the number of edge rewires since the perturbed graph has similar `shape' as the original one. To confirm this intuition, we depart from existing analysis which rely on properties of polynomial graph filter. Instead, we treat the graph filter distance using analysis in the graph frequency domain.
Assuming that the graph filter is \emph{low pass} \cite{ramakrishna} such that the filter suppresses high graph frequency components of its input, our analysis relates the stability of graph filter to changes in community structure of graph topology. Finally, we yield a graph distance bound that can be \emph{invariant} to the number of edge rewires in the graph perturbation. Our contributions are:\vspace{-.1cm}
\begin{itemize}[leftmargin=4mm, noitemsep]
\item We derive the first \emph{community structure} dependent bound on graph filter distance which accounts for the perturbation to the graph topology in terms of community structure and the low pass filtering capability of graph filter, see Theorem~\ref{thm:1}. 
\item Our bound is specialized to the class of graph perturbation generated from a homogeneous stochastic block model. When the graph filter is sufficiently low pass, we show that the \emph{graph filter distance approaches zero} as the number of nodes grows, see Corollaries~\ref{cor:1} and \ref{cor:2}. We emphasize that our result is applicable to \emph{unnormalized Laplacian} as GSO, which is not covered by most of the previous works on stability of graph filters. \vspace{-.1cm}
\end{itemize}
Lastly, numerical simulations on synthetic and real graphs are provided to verify our findings.\vspace{.1cm}

\noindent \textbf{Related Works}. The community structure of a graph emerges when the number of nodes is large. As such, our study is closely related to works which study the convergence properties of graph filters with very large graphs. For instance, \cite{ruiz2021graphon} studied the generalization of graph signal processing (GSP) \cite{ortega2018graph} to graphons, \cite{ruiz2020graphon,keriven2020convergence} analyzed the stability of GCNs with an infinite number of nodes/data. We remark that our Corollary~\ref{cor:1} can be viewed as an extension of \cite{keriven2020convergence} to the unnormalized Laplacian as GSO.



\vspace{.1cm}
\noindent \textbf{Notations}. We use $||\cdot ||_2$ to denote spectral norm for matrices and Euclidean norm for vectors. For a vector (resp.~matrix) ${\bf x}$ (resp.~${\bf X}$), we use $[ {\bf x} ]_i$ (resp.~$[ {\bf X} ]_{ij}$) to denote its $i$th (resp.~$(i,j)$th) entry.\vspace{-.2cm}

\section{Preliminaries}\vspace{-.2cm} \label{sec:prelim}

\noindent \textbf{Graphs and Graph Signals}. We consider undirected and connected graphs $\mathcal{G} = (\mathcal{V}, \mathcal{E})$, where $\mathcal{V}$ is set of $n$ nodes and $\mathcal{E}$ is the edge set. The graph ${\cal G}$ is endowed with a graph shift operator (GSO), which is defined as a symmetric matrix $\mathbf{S} \in \mathbb{R}^{n \times n}$ such that $[\mathbf{S}]_{ij} \neq 0$ if and only if $i = j$ or $(i, j) \in \mathcal{E}$ \cite{ortega2018graph}. We focus on two common GSOs: unnormalized Laplacian, normalized Laplacian. Given the adjacency matrix $\mathbf{A} \in \mathbb{R}^{n \times n}$ and the degree matrix $\mathbf{D} = {\rm Diag} (\mathbf{A}\mathbf{1})$, the unnormalized Laplacian is defined as $\mathbf{L}_U = \mathbf{D} - \mathbf{A}$ and the normalized Laplacian is defined as $\mathbf{L}_{\sf norm} = \mathbf{D}^{-1/2}\mathbf{L}_U\mathbf{D}^{-1/2}$.
The GSO ${\bf S}$ admits an eigendecomposition as $\mathbf{S} = \mathbf{V}\mathbf{\Lambda}\mathbf{V}^\top$, where $\mathbf{V} = (\mathbf{v}_1, ..., \mathbf{v}_n) \in \mathbb{R}^{n \times n}$ collects the orthonormal eigenvectors and $\mathbf{\Lambda} = {\rm Diag}(\lambda_1, ..., \lambda_n)$ is a diagonal matrix of the eigenvalues.
Throughout this paper, we order the eigenvalues such that $0 = \lambda_1 < \lambda_2 \leq \cdots \leq \lambda_n$.

The eigenvalues of ${\bf S}$ can be interpreted as the \emph{graph frequencies} with the smaller eigenvalues correspond to lower frequencies. This can be understood with the graph total variation ${\bf x}^\top {\bf S} {\bf x}/2 = \sum_{ (i,j) \in {\cal E} } ( x_i - x_j )^2$ for the graph signal on ${\cal G}$ represented by ${\bf x} \in \mathbb{R}^n$. As ${\bf v}_i^\top {\bf S} {\bf v}_i = \lambda_i$, we observe that the eigenvector ${\bf v}_i$ has higher variations with respect to ${\cal G}$ as the frequency index $i$ increases. We define the \emph{graph Fourier transform} (GFT) of ${\bf x}$ as $\widehat{\bf x} = {\bf V}^\top {\bf x}$ such that $[\widehat{\bf x}]_i$ denotes the magnitude of the $i$th frequency component.\vspace{.1cm} 

\noindent \textbf{Low Pass Graph Filters}. A graph filter $\mathcal{H}(\mathbf{S}) \in \mathbb{R}^{n \times n}$ maps the input signal $\mathbf{x} \in \mathbb{R}^n$ to the output signal $\mathbf{y} = \mathcal{H}(\mathbf{S})\mathbf{x} \in \mathbb{R}^{n}$. Specifically, we consider graph filter of the form:\vspace{-.1cm}
\begin{equation} \textstyle
    \mathcal{H}(\mathbf{S}) = \sum_{t=0}^{T-1}h_t \mathbf{S}^t.\vspace{-.1cm}
\end{equation}
Note that $T \in \mathbb{N}$ is the filter's order. 
Based on the eigendecomposition of GSO $\mathbf{S} = \mathbf{V}\mathbf{\Lambda}\mathbf{V}^\top$, we define the frequency responses as $h(\lambda) := \sum_{t=0}^{T-1} h_t\lambda^t $ and subsequently the graph filter can be written as $\mathcal{H}(\mathbf{S}) = \mathbf{V}h(\mathbf{\Lambda})\mathbf{V}^\top$, where $h(\mathbf{\Lambda}) = {\rm Diag}\left(h(\lambda_1), ..., h(\lambda_n)\right)$.
We note that the GFT of ${\bf y}$ can be written as $\widehat{\bf y} = h ( \bm{\Lambda} ) \widehat{\bf x}$.

This study is concerned with ${\cal H}(\cdot)$ that are \emph{low pass graph filters} which retain the low frequency components of its input while suppressing the high frequency components. We define this class of graph filters of interest with respect to the frequency response function $h(\lambda)$. Our definition is extended from \cite{ramakrishna} as follows:\vspace{-.1cm}
\begin{Def} \label{def:1}
The graph filter $\mathcal{H}(\cdot)$ is said to be $\Bar{\lambda}$-low pass if\vspace{-.1cm}
\begin{equation}
    \eta = \big( \min\limits_{\lambda \in \left[0, \Bar{\lambda}\right]}\left|h(\lambda)\right| \big)^{-1} {\max\limits_{\lambda \in \left[\Bar{\lambda}, \infty \right)}\left|h(\lambda)\right|} < 1,\vspace{-.2cm}
\end{equation}
where $\eta < 1$ is known as the low pass ratio for the graph filter.\vspace{-.1cm}
\end{Def}
\noindent The above definition requires the graph filter ${\cal H}(\cdot)$ to have a cutoff frequency at $\Bar{\lambda}$ such that any components with graph frequency above $\Bar{\lambda}$ are attenuated by a factor of at least $\eta < 1$. Unlike \cite{ramakrishna}, our definition is applicable regardless of the GSO ${\bf S}$ used as it only constraints on the form of the graph filter function determined by the filter coefficients $\{ h_t \}_{t=0}^{T-1}$. Lastly, the cutoff frequency $\Bar{\lambda}$ maybe different depending on the graph size $n$ to achieve a reasonable low pass filtering performance.
\vspace{.1cm}

\noindent \textbf{Stability of Graph Filters}. 
We consider a perturbation of ${\cal G}$ into $\hat{\cal G}$, e.g., via edge rewiring, which results in the GSO $\hat{\bf S}$. Our goal is to study the \emph{stability of ${\cal H}({\bf S})$} via analyzing the amount of changes in the filter's output under perturbation. 
As inspired by \cite{gama}, we measure the \emph{stability of ${\cal H}({\bf S})$} using the quantity:
\begin{align}
\label{graphfilterdistance}
\hspace{-.1cm}\mathbb{D}_{\cal H} ( {\bf S}, \hat{\bf S} ) = \hspace{-.4cm} \sup_{ {\bf x} \in \mathbb{R}^n, {\bf x} \neq {\bf 0} } \hspace{-.4cm}
\frac{\big\|\mathcal{H}(\mathbf{S})\mathbf{x} - \mathcal{H}(\mathbf{\hat{S}})\mathbf{x} \big\|_2}{||\mathbf{x}||_2} = \| \mathcal{H}(\mathbf{S}) - \mathcal{H}( \hat{\bf S} ) \|_2.\hspace{-.1cm}
\end{align}
We call the quantity $\mathbb{D}_{\cal H} ( {\bf S}, \hat{\bf S} )$ the \textit{graph filter distance} between the filters $\mathcal{H}(\mathbf{S})$ and $\mathcal{H}(\mathbf{\hat{S}})$. Notice that \eqref{graphfilterdistance} yields an upper bound to the operator distance modulo permutation measure $\| \cdot \|_{\cal P}$ defined in \cite{gama} since we have ignored the ambiguity due to node permutations.\vspace{.1cm}

\noindent \textbf{Stochastic Block Model}. \label{subsec:sbm} To derive stability properties of graph filters under a large number of edge rewires, we concentrate on the stochastic block models (SBMs) random graphs to develop insights. Here, ${\sf SBM}(n, k, {\bf B}, {\bf Z})$ \cite{rohe} denotes a random graph model with $n$ nodes partitioned into $k$ blocks, the membership matrix $\mathbf{Z} \in \{0, 1\}^{n \times k}$ such that $[\mathbf{Z}]_{ij} = 1$ if and only if node $i$ is in community $j$, and a connectivity matrix $\mathbf{B} = [b_{ij}]_{1 \leq i, j \leq k} \in [0, 1]^{k \times k}$, whose entries $b_{ij}$ being the probability of edges between nodes in block $i$ and block $j$. Let ${\cal G} \sim {\sf SBM}(n, k, {\bf B}, {\bf Z})$, its empirical adjacency matrix $\mathbf{A}$ satisfies $\mathbb{E}[\mathbf{A}] = \mathcal{A} = \mathbf{Z}\mathbf{B}\mathbf{Z}^\top$, which is referred to as the population adjacency matrix. 
Lastly, we also consider a special case of SBMs given by the planted partition model ${\sf PPM}(n, k, a_n, b_n)$ which has $k$ equally sized blocks with $n/k$ nodes per block. The PPM is a special case of an SBM where the connectivity matrix is given by $\mathbf{B} = a_n\mathbf{I} + b_n\mathbf{1}\mathbf{1}^T$. 
\vspace{-.2cm}

\section{Main results}\label{sec:results}\vspace{-.2cm}
This section presents our main findings on the stability property of low pass graph filters. We first present a bound on $\mathbb{D}_{\cal H} ( {\bf S}, \hat{\bf S} )$ which explicitly considers the strength of low pass filtering and the structural difference between ${\bf S}, \hat{\bf S}$. 
Let us introduce two assumptions for simplifying the constants in our bound:\vspace{-.1cm}
\begin{assumption}[$\Bar{\lambda}$] \label{assumption:1} There exists a constant $\mathbb{H}_{\max}$ such that\vspace{-.2cm}
\begin{equation}
    \textstyle \sup_{\lambda \in [0, \Bar{\lambda}]} |h(\lambda)| \leq \mathbb{H}_{\max}.\vspace{-.2cm}
\end{equation}
\end{assumption}
\begin{assumption}[$\Bar{\lambda}$]
\label{assumption:2}
There exists a constant $L_{\mathbb{H}}$ such that\vspace{-.2cm}
\begin{equation}
    |h(\lambda) - h(\lambda')| \leq L_{\mathbb{H}} \, |\lambda - \lambda'|,~\forall~\lambda, \lambda' \in [0, \Bar{\lambda}].\vspace{-.1cm}
\end{equation}
\end{assumption}
\noindent H\ref{assumption:1}  assumes that the frequency response function of the graph filter is bounded in the interval $[0, \Bar{\lambda}]$. Meanwhile, H\ref{assumption:2} imposes a Lipschitz continuity condition on the frequency response function. Note that H\ref{assumption:2} is weaker than a uniform Lipschitz filter condition as we only consider the graph frequencies up to $\lambda \leq \Bar{\lambda}$, and when $\Bar{\lambda}$ is small, it can imply the integral Lipschitz condition in \cite{gama}, i.e., $\sup_{\lambda \geq 0} | \lambda h'(\lambda) | \leq C$. An example  satisfying Definition~\ref{def:1}, H\ref{assumption:1}, H\ref{assumption:2} is the exponential graph filter ${\cal H}( {\bf S} ) = e^{- \sigma {\bf S}}$, where $\sigma > 0$.

The following theorem bounds the graph filter distance $\mathbb{D}_{\cal H}({\bf S}, \hat{\bf S})$. For brevity, we denote $\mathbf{\Lambda}_k$ and $\mathbf{V}_k$ (resp. $\mathbf{\hat{\Lambda}}_k$ and $\mathbf{\hat{V}}_k$) as the matrices of smallest-$k$ eigenvalues and eigenvectors of $\mathbf{S}$ (resp.~$\hat{\bf S}$):
\begin{Theorem}
\label{thm:1}
Let ${\cal H}(\cdot)$ be a $\Bar{\lambda}$-low pass filter with ratio $\eta$. Suppose that H\ref{assumption:1}($\Bar{\lambda}$), H\ref{assumption:2}($\Bar{\lambda}$) hold, and the graph frequencies of the GSOs ${\bf S}, \hat{\bf S}$ satisfy $\lambda_k \leq \Bar{\lambda} \leq \lambda_{k+1}$ and $\hat{\lambda}_k \leq \Bar{\lambda} \leq \hat{\lambda}_{k+1}$ for some $1 \leq k \leq n-1$. Then, the graph filter distance satisfies:\vspace{-.1cm}
\begin{align}
\mathbb{D}_{\cal H}({\bf S}, \hat{\bf S}) & \leq 2 \mathbb{H}_{\max} \, \eta \label{eq:gfbd} \\
& \quad + L_{\mathbb{H}} ||\mathbf{\Lambda}_k - \mathbf{\hat{\Lambda}}_k||_2 + 2 \mathbb{H}_{\max} ||\mathbf{V}_k - \mathbf{\hat{V}}_k||_2 . \nonumber
\end{align}
\end{Theorem}
\begin{proof}
Using the triangular inequality, we observe that\vspace{-.1cm}
\begin{align}
\mathbb{D}_{\cal H}({\bf S}, \hat{\bf S})   &\leq \|\mathbf{V}_k h(\mathbf{\Lambda}_k) \mathbf{V}_k^T - \mathbf{\hat{V}}_k h(\mathbf{\hat{\Lambda}}_k) \mathbf{\hat{V}}_k \|_2 \nonumber \\
& \textstyle + \big\|\sum_{i=k+1}^{n} \{ h(\lambda_i)\mathbf{v}_i\mathbf{v}_i^\top - h(\hat{\lambda}_i)\mathbf{\hat{v}}_i\mathbf{\hat{v}}_i^\top \} \big\|_2. \label{expr:3}
\end{align}
By H\ref{assumption:1}($\Bar{\lambda}$) and the property of $\Bar{\lambda}$-low pass graph filters as specified in Definition \ref{def:1}, it is easy to derive that \vspace{-.1cm}
\begin{equation} \label{eq:maxeta}
\textstyle \max_{i \in \{k+1, ..., n\}} \max\{ |h(\lambda_i)|, |h(\hat{\lambda}_i)| \} \leq \eta \, \mathbb{H}_{\max}.
\end{equation}
This bounds the second term in \eqref{expr:3} as $2 \eta \, \mathbb{H}_{\max}$ where we have further used $\| \sum_{i=k+1}^n {\bf v}_i {\bf v}_i^\top \|_2 \leq 1$. On the other hand, by H\ref{assumption:1}($\Bar{\lambda}$), H\ref{assumption:2}($\Bar{\lambda}$), we can upper bound the first term in \eqref{expr:3} as \vspace{-.1cm}
\begin{align}
&||\mathbf{V}_k h(\mathbf{\Lambda}_k) \mathbf{V}_k^\top - \mathbf{\hat{V}}_k h(\mathbf{\hat{\Lambda}}_k) \mathbf{\hat{V}}_k^\top ||_2 \nonumber \\
&\leq || h(\mathbf{\Lambda}_k) - h( {\bm{\hat{\Lambda}}_k}) ||_2 + ||\mathbf{V}_k h(\mathbf{\hat{\Lambda}}_k) \mathbf{V}_k - \mathbf{\hat{V}}_k h(\mathbf{\hat{\Lambda}}_k) \mathbf{\hat{V}}_k||_2 \nonumber \\
& \leq L_{\mathbb{H}} \| \bm{\Lambda}_k - \bm{\hat{\Lambda}}_k \|_2 + 2 \mathbb{H}_{\max} \| {\bf V}_k - \hat{\bf V}_k \|_2 \label{eq:lap}
\end{align}
Combining \eqref{eq:maxeta}, \eqref{eq:lap} yields the desired bound.\vspace{-.1cm}
\end{proof}
\noindent The condition which requires the cutoff frequency $\Bar{\lambda}$ to satisfy $\lambda_k \leq \Bar{\lambda} \leq \lambda_{k+1}$ and $\hat{\lambda}_k \leq \Bar{\lambda} \leq \hat{\lambda}_{k+1}$ impose a \emph{spectral gap} requirement on the eigenvalues of ${\bf S}, \hat{\bf S}$ where a common $\Bar{\lambda}$ separates the `low' and `high' frequencies of the two graphs. This condition can be satisfied if the two graphs ${\bf S}, \hat{\bf S}$ are $k$-modular \cite{newman2013spectral}. We remark that ${\cal H}({\bf S})$, ${\cal H}(\hat{\bf S})$ are also $k$-low pass graph filter \cite{ramakrishna} in this case.

The graph filter distance bound \eqref{eq:gfbd} comprises of three terms. The first depends on the low pass ratio $\eta$ which relies on the graph filter's frequency response. The last two terms capture the similarity in the graph structure by comparing the pairs ${\bf V}_k, \hat{\bf V}_k$ and $\bm{\Lambda}_k, \bm{\hat{\Lambda}}_k$. We expect the pairs $\| {\bf V}_k - \hat{\bf V}_k \|_2$ and $\| \bm{\Lambda}_k - \bm{\hat{\Lambda}}_k\|_2$ to be small when ${\bf S}, \hat{\bf S}$ are similar in terms of their community structures. 

Theorem~\ref{thm:1} implies that with a low pass filter ${\cal H}(\cdot)$ satisfying $\eta \ll 1$, the graph filter distance only depends on the difference between ${\bf S}, \hat{\bf S}$ in terms of their \emph{community structure}. 
The distance will be insensitive to the \emph{number of edge rewiring} unlike in \cite{gama,kenlay_icassp21}. Our next endeavor is to showcase examples of graph perturbation where the community structures are invariant, thereby applying Theorem~\ref{thm:1} concludes that the graph filter is stable.\vspace{.2cm}

\noindent \textbf{Edge Rewiring Scheme.} \label{subsec:edgerewiring} We focus on a simple case based on SBM by fixing the number of blocks $k$, the membership matrix ${\bf Z}$ and connectivity matrix ${\bf B}$. In this example, both the original graph ${\cal G}$ and the perturbed graph $\hat{\cal G}$ are generated from the SBM with the same parameters such that ${\cal G}, \hat{\cal G} \sim {\sf SBM} (n,k, {\bf B}, {\bf Z})$. 
The above setting covers an edge rewiring scheme as follows: 
\begin{enumerate}[noitemsep,leftmargin=4mm]
\item Generate the original graph as ${\cal G} \sim {\sf SBM} (n,k, {\bf B}, {\bf Z})$.
\item For each inter/intra-cluster block $(i,j)$, {\sf (i)} delete a portion of $p_{\sf re} \in [0,1]$ edges uniformly, and {\sf (ii)} add edges to the node pairs without edges with probability $[ b_{ij}^{-1} - (1 - p_{\sf re} ) ]^{-1} p_{\sf re}$, selected independently.
\end{enumerate} 
Notice that edge rewiring in the above only occurs for the edges within the same inter/intra-cluster block. It can be shown that the resultant perturbed graph satisfies $\hat{\cal G} \sim {\sf SBM} (n,k, {\bf B}, {\bf Z})$.\vspace{-.2cm}

\subsection{Stability with {Unnormalized Laplacian as GSO}}\vspace{-.1cm}
Our plan is to study the graph filter distance bound through borrowing recent consistency result of SBMs when $n \rightarrow \infty$. Notice that works on the convergence of \emph{unnormalized Laplacian} are scarce. This may be due to the fact that unnormalized Laplacian does not have good concentration due to the involved diagonal degree matrix \cite{deng}, \cite{luxberg_08}. Hence, we cannot use known perturbation bounds such as Weyl's inequality \cite{bhatia} and Davis-Kahan theorem \cite{yu} to derive concentration of eigenvalues and eigenvectors. In the following, we adopt the results from \cite{deng} and concentrate on a special case of the SBM model given by a \emph{sparse} PPM with $k=2$ blocks such that ${\cal G} \sim {\sf PPM}(n, 2, \alpha \log n/n, \beta \log n/n)$, where $\alpha, \beta \in \mathbb{R}^+$.  

We first observe the following high probability bounds on the bottom-$2$ eigenvectors and eigenvalues:
\begin{Corollary} \label{cor:1}
Let $\alpha, \beta \in \mathbb{R}^{+}$, consider PPM graphs as ${\cal G}, \hat{\cal G} \sim {\sf PPM}(n, 2, \alpha \log n/n, \beta \log n/n)$.
Denote their unnormalized Laplacian as ${\bf L}_U$, $\hat{\bf L}_U$. Moreover, $\mathbf{\Lambda}_2$ and $\mathbf{V}_2$ (resp. $\mathbf{\hat{\Lambda}}_2$ and $\mathbf{\hat{V}}_2$) are the matrices of smallest-$2$ eigenvalues and eigenvectors of $\mathbf{L}_U$ (resp.~$\hat{\bf L}_U$). Suppose that $\sqrt{\alpha} - \sqrt{\beta} > \sqrt{2}$, then with probability at least $1-o(1)$, it holds \vspace{-.1cm}
\begin{equation}
   \| {\bf V}_2 - \hat{\bf V}_2 \|_2 = o(1), \quad \| \bm{\Lambda}_2 - \bm{\hat{\Lambda}}_2 \|_2 = {\cal O}( \log n / n ), \vspace{-.1cm}
\end{equation}
where $o(1) \rightarrow 0$ as $n \rightarrow \infty$.
\end{Corollary}
\noindent The corollary is obtained by Theorem 8, Lemma 9 and Theorem 11 of \cite{deng} and the randomness is due to the generation of PPM graphs. The condition $\sqrt{\alpha} - \sqrt{\beta} > \sqrt{2}$ is known as the spectral gap criterion which is necessary to distinguish the blocks in PPM \cite[Theorem 13]{abbe}. We shall skip the proof in the interest of space. 

To apply Theorem~\ref{thm:1}, we choose $k=2$ as the PPM contains 2 blocks which results in a graph with $2$ densely connected clusters. We observe that for the unnormalized Laplacian of sparse PPM graph, its non-zero eigenvalues grow with $\Theta( \log n )$. Satisfying the spectral gap condition $\lambda_2 \leq \Bar{\lambda} \leq \lambda_3$, $\hat\lambda_2 \leq \Bar{\lambda} \leq \hat\lambda_3$ therefore requires the cutoff frequency $\Bar{\lambda}$ in Definition~\ref{def:1} to grow with $\log n$ as well. To this end, if we consider the exponential graph filter\footnote{In fact, this holds even for the filter ${\cal H}( {\bf L}_U ) = e^{- \sigma {\bf L}_U / \log n }$.} ${\cal H}( {\bf L}_U ) = e^{- \sigma {\bf L}_U }$, we observe that the low pass ratio $\eta$, the constants in H\ref{assumption:1}($\Bar{\lambda}$), H\ref{assumption:2}($\Bar{\lambda}$) are insensitive to the growth of $\Bar{\lambda} = \Theta(\log n)$. Under the above premises, we have \vspace{-.1cm}
\begin{equation}
    \mathbb{D}_{\cal H} ( {\bf L}_U, \hat{\bf L}_U ) \leq 2 \eta \, \mathbb{H}_{\max} + \mathbb{H}_{\max} \, o( 1 ) + L_{\mathbb{H}} \, {\cal O} ( \log n / n ), \vspace{-.1cm}
\end{equation}
with high probability. As such, the graph filter distance between ${\cal H}({\bf L}_U)$ and ${\cal H}(\hat{\bf L}_U)$ is small when $\eta \ll 1$ and $n \rightarrow \infty$.\vspace{-.1cm}


\subsection{Stability with {Normalized Laplacian as GSO}}\vspace{-.1cm}
Similar to the previous subsection, our plan is to borrow the classical consistency result on SBMs from \cite{deng} and show that the low pass graph filter is stable as $n \rightarrow \infty$. Let $\alpha, \beta \in \mathbb{R}^+$, we consider the $k$-blocks sparse PPM with ${\cal G} \sim {\sf PPM}(n, k, \alpha \log n / n, \beta \log n / n)$, we observe the following high probability bound on the difference between the original and perturbed normalized Laplacian:
\begin{Corollary} \label{cor:2}
Let $\alpha, \beta \in \mathbb{R}^+$, consider PPM graphs as ${\cal G}, \hat{\cal G} \sim {\sf PPM}(n, k, \alpha \log n / n, \beta \log n / n)$.
Denote their normalized Laplacian as ${\bf L}_{\sf norm}$, $\hat{\bf L}_{\sf norm}$. 
With probability at least $1-o(1)$, it holds\vspace{-.1cm}
\begin{align}
\| {\bf L}_{\sf norm} - \hat{\bf L}_{\sf norm} \|_2 = {\cal O}\left( 1/\sqrt{\log n} \right).
\end{align}
\end{Corollary}
\noindent The above corollary is obtained from extending \cite[Theorem 4]{deng}, which is based on \cite{lei}. In particular, in the PPM graphs considered, we have  $\min_{ij}{\cal A}_{ij} \geq c_0\log n / n$ such that the minimum degrees of ${\cal G}, \hat{\cal G}$ are $\Omega( \log n )$. We skip the proof in the interest of space.
Notice that compared to the previous case with unnormalized Laplacian matrices, we obtained a stronger concentration with respect to the normalized Laplacian matrices themselves.


To apply Theorem~\ref{thm:1}, using the Weyl's inequality and Davis-Kahan theorem one can show that $\| {\bf V}_k - \hat{\bf V}_k \|_2 = {\cal O}( 1/\sqrt{ \log n })$, $\| \bm{\Lambda}_k - \bm{\hat\Lambda}_k \|_2 = {\cal O}( 1/\sqrt{ \log n })$ for the bottom-$k$ eigenvectors and eigenvalues. Furthermore, if we consider the special case of $k = 2$ blocks PPM graphs, then $\lambda_2 ( {\bf L}_{\sf norm} )$ will be approximately  $2 \beta / (\alpha + \beta)$ while $\lambda_3 ( {\bf L}_{\sf norm} )$ is at least $1 - {\cal O}(1/\sqrt{\log n})$ \cite[Theorem 10]{deng}. With the spectral gap condition $\sqrt{\alpha} - \sqrt{\beta} > \sqrt{2}$ and let $\beta \leq 1$ for simplicity, as long as the graph filter ${\cal H}(\cdot)$ satisfies Definition~\ref{def:1} with a constant cutoff frequency $\Bar{\lambda} \approx 1/2$, we obtain:\vspace{-.1cm}
\begin{equation}
    \mathbb{D}_{\cal H} ( {\bf L}_{\sf norm}, \hat{\bf L}_{\sf norm} ) \leq 2 \eta \, \mathbb{H}_{\max} + (\mathbb{H}_{\max} + L_{\mathbb{H}}) \, {\cal O}\left( 1/\sqrt{\log n} \right), \vspace{-.1cm}
\end{equation}
with high probability. In other words, the graph filter distance between ${\cal H}({\bf L}_{\sf norm})$, ${\cal H}(\hat{\bf L}_{\sf norm})$ is small when $\eta \ll 1$ and $n \rightarrow \infty$.

\begin{Remark} One may apply \cite[Proposition 1]{kenlay_icml21} to yield \vspace{-.1cm}
\begin{equation} \label{eq:stable_kenlay}
    \textstyle 
    \mathbb{D}_{\cal H} ( {\bf L}_{\sf norm}, \hat{\bf L}_{\sf norm} ) \leq \sum_{t=0}^{T-1} t \, 2^{t-1} \, |h_t| \, \| {\bf L}_{\sf norm} - \hat{\bf L}_{\sf norm} \|_2, \vspace{-.1cm}
\end{equation}
where we have used $\| {\bf L}_{\sf norm} \|_2 \leq 2$. Therefore, we observe that when $T < \infty$, \eqref{eq:stable_kenlay} also yields the stability property for the graph filter as $n \rightarrow \infty$, without using the low pass condition on ${\cal H}(\cdot)$. In comparison, our result allows the filter order $T$ to be infinite. 
Note \cite{kenlay_icml21} also considered a special case with ${\cal H}({\bf L}) = ( {\bf I} + \alpha {\bf L})^{-1}$.\vspace{-.2cm} 
\end{Remark}

\section{Numerical experiments}\label{sec:exp}\vspace{-.2cm}
\noindent \textbf{Synthetic Experiment.} We compare the graph filter distances $\mathbb{D}_{\cal H} ( {\bf S}, \hat{\bf S} )$ [cf.~\eqref{graphfilterdistance}] under different configurations of graph filters and GSOs, as summarized in Table~\ref{tab:filtersettings}. Notice that in addition to the \emph{low pass} graph filters ${\cal H}_{\sf LP}({\bf S})$ which satisfy Definition~\ref{def:1}, as a control experiment we also consider \emph{high pass} graph filters ${\cal H}_{\sf HP}({\bf S})$. The latter aims to illustrate if the low pass property is necessary for the stability of graph filter in both situation. 

\begin{table}[t]
\caption{Graph filter settings for synthetic experiments}\vspace{-.2cm} \label{tab:filtersettings} 
\centering
\begin{tabular}{ l l l }
 \toprule
 & \text{Unnormalized Laplacian} & \text{Normalized Laplacian} \\ 
 \midrule
 ${\cal H}_{\sf LP}(\cdot)$ & $\exp(-(1/\log n)\mathbf{L}_U)$ & $\exp(-\mathbf{L}_{\sf norm})$ \\
 \hline
 ${\cal H}_{\sf HP}(\cdot)$ & $\exp((1/\log n)\mathbf{L}_U)$ & $\exp(\mathbf{L}_{\sf norm})$ \\ 
 \bottomrule 
\end{tabular}
\end{table}

\pgfplotsset{every tick label/.append style={font=\small}}
\begin{figure}[t]
\begin{subfigure}[b]{1\linewidth}
\centering
  {\sf \resizebox{1.\linewidth}{!}{\input{fig_D_n_1}}}
  \caption{Unnormalized Laplacian ${\bf L}_U$ as GSO.}\vspace{-.2cm}
  \label{fig:D_n_1}
\end{subfigure}
\newline
\newline
\begin{subfigure}[b]{1\linewidth}
\centering
  {\sf \resizebox{1\linewidth}{!}{\input{fig_D_n_2}}}
  \caption{Normalized Laplacian ${\bf L}_{\sf norm}$ as GSO.}\vspace{-.2cm}
  \label{fig:D_n_2}
\end{subfigure}
\caption{Comparing the average graph filter distance $\mathbb{D}_{\cal H} ( {\bf S}, \hat{\bf S} )$ against the number of nodes $n$. (Left) High pass filter ${\cal H}_{\sf HP}(\cdot)$ (Right) Low pass filter ${\cal H}_{\sf LP}(\cdot)$. The vertical bars indicate $95\%$ confidence intervals of the graph filter distance.}\vspace{-.3cm}
\label{fig:D_n}
\end{figure}
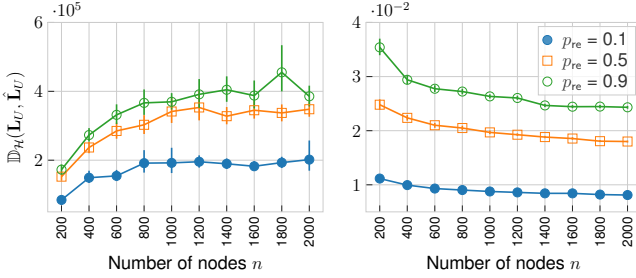
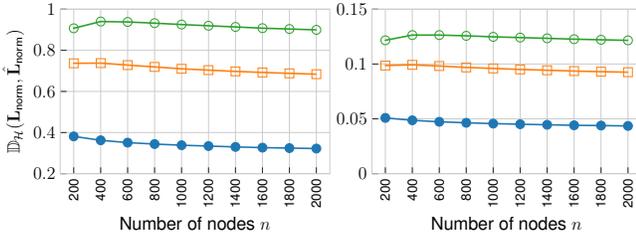

To simulate the edge rewiring graph perturbations, we generate the original graphs ${\cal G}$ from a $k=2$ blocks PPM with ${\cal G} \sim {\sf PPM}( n, 2, 13\log n/n, 2\log n/n)$ with $n \in [200, 2000]$. An intra-block edge is assigned with probability $15\log n/n$ while an inter-block edge is assigned with probability $2\log n/n$. The perturbed graph is generated from ${\cal G}$ by the edge rewiring process in Section~\ref{subsec:edgerewiring}. We perform Monte-Carlo simulations with $100$ trials to estimate $\mathbb{E}[ {\cal D}_{\cal H}( {\bf S}, \hat{\bf S} )]$ under the random graph model and its perturbations. 

\begin{figure}[t]
    \pgfplotsset{every tick label/.append style={font=\tiny}}
    \pgfplotsset{every axis/.append style={ label style={font=\tiny}}}
    {\sf \resizebox{.53\linewidth}{!}{\input{fig_D_p_strong}}~\resizebox{.47\linewidth}{!}{\input{fig_D_p_weak}}}\vspace{-.2cm}
    \caption{Comparing the averaged graph filter distance $\mathbb{D}_{\cal H} ( {\bf L}_U, \hat{\bf L}_U )$ against rewiring ratio $p_{\sf re}$ over $100$ trials for different graph filters on the \texttt{email-Eu-core} dataset. Note that the red (resp.~green) curves correspond to low pass (resp.~high pass) graph filters.
    }\vspace{-.3cm}
    \label{fig:D_p_real}
\end{figure}
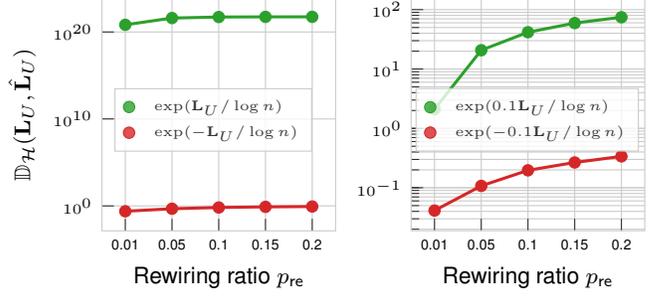

The simulation results can be found in Fig.~\ref{fig:D_n}. In the figure, we compare the averaged graph filter distance for the four settings in Table~\ref{tab:filtersettings}. Particularly, we consider three rewiring ratios $p_{\sf re} \in \{ 0.1, 0.5, 0.9 \}$ and make the following observations in order. 
In the case when unnormalized Laplacian is taken as the GSO, we observe that the low pass graph filters stabilize with ${\cal D}_{ {\cal H}_{\sf LP} }( {\bf L}_U, \hat{\bf L}_U ) \rightarrow 0$ as $n \rightarrow \infty$; in contrast, the high pass graph filters become unstable as the graph filter distance increases rapidly with $n$, note that ${\cal D}_{ {\cal H}_{\sf HP} }( {\bf L}_U, \hat{\bf L}_U ) \approx 10^5$.
In the case when normalized Laplacian is taken as the GSO, both low pass and high pass graph filters stabilize as $n \rightarrow \infty$. As observed, the convergence rate of graph filter distance with low pass filters is slightly faster than the high pass filters.

The above results agree with Section~\ref{sec:results} as the low pass graph filters stabilize as $n \rightarrow \infty$. Moreover, our observations seem to suggest that stability with large number of rewires does not hold for high pass filter with ${\bf L}_U$ as GSO. On the other hand, for high pass graph filters with ${\bf L}_{\sf norm}$, we suspect that the convergence of ${\cal D}_{\cal H_{\sf HP}} ( {\bf L}_{\sf norm}, \hat{\bf L}_{\sf norm} )$ is due to the discussions in \eqref{eq:stable_kenlay}, where the facts $\| {\bf L}_{\sf norm} \|_2 \leq 2$ and $\| {\bf L}_{\sf norm} - \hat{\bf L}_{\sf norm} \|_2 \rightarrow 0$ suffice to ensure stability.\vspace{.2cm}

\noindent \textbf{Real Data Experiment.} Our last example considers the stability of graph filters on real graph topology. We consider ${\cal G}$ as the \texttt{email-Eu-core} network \cite{leskovec}. The latter is an email network containing $1,005$ nodes with $25,571$ edges and $42$ communities representing the departmental membership of researchers. To test the stability of a graph filter defined on ${\cal G}$, we perform a simplified edge rewiring process: for each inter and intra-cluster block, {\sf (i)} we delete a portion $p_{\sf re} \in [0,1]$ of edges selected uniformly at random, then {\sf (ii)} we add back the same number of edges into the respective block, again selected uniformly at random. 

Fig.~\ref{fig:D_p_real} shows the graph filter distance against the rewiring ratio $p_{\sf re}$ for two pairs of low pass and high pass graph filters with unnormalized Laplacians as GSO. In both cases, we observe that low pass filters are stable over the considered range of rewiring ratios. Especially with a set of `weaker' low/high pass filter [cf.~right panel with ${\cal H}( {\bf L}_U ) = \exp( \pm 0.1 {\bf L}_U / \log n )$], we observe that $\mathbb{D}_{\cal H}( {\bf L}_U, \hat{\bf L}_U )$ increases at slower rate for the low pass filter. 
\vspace{.1cm}

\noindent \textbf{Conclusions.}
We study the stability of low pass graph filters when subjected to a large number of edge rewires. We propose a new stability bound given with respect to the frequency response of filters. Our bound shows that the stability property for low pass graph filters hinges on whether there are changes to the community structure in the perturbation. Numerical experiments validate our theories.

\newpage

\bibliographystyle{IEEEtran}
\bibliography{ref}


\end{document}

%% file: fig_D_n_1.tex
\begin{tikzpicture}

\definecolor{color0}{rgb}{0.12156862745098,0.466666666666667,0.705882352941177}
\definecolor{color1}{rgb}{1,0.498039215686275,0.0549019607843137}
\definecolor{color2}{rgb}{0.172549019607843,0.627450980392157,0.172549019607843}

\begin{groupplot}[group style={group size=2 by 1}]
\nextgroupplot[
axis line style={white!80!black},
legend cell align={left},
legend style={
  fill opacity=0.8,
  draw opacity=1,
  text opacity=1,
  at={(0.03,0.97)},
  anchor=north west,
  draw=white!80!black,
  font = \large
},
log basis y={10},
tick align=inside,
x grid style={white!80!black},
xlabel={\large Number of nodes \(\displaystyle n\)},
xmajorgrids,
xmajorticks=true,
xmin=-0.5, xmax=9.5,
xminorgrids,
xtick style={color=white!15!black},
xtick={0,1,2,3,4,5,6,7,8,9},
xticklabels={200,400,600,800,1000,1200,1400,1600,1800,2000},
xticklabel style={rotate=90},
xtick pos=left,
y grid style={white!80!black},
ylabel={\large \(\displaystyle \mathbb{D}_{\cal H} ( {\bf L}_U, \hat{\bf L}_U )\)},
ymajorgrids,
ymajorticks=true,
ymin=50000, ymax=600000,
yminorgrids,
ytick style={color=white!15!black},
yticklabel style={font=\large},
ytick pos=left,
width = 8cm,
height = 6cm
]
\addplot [draw=color0, fill=color0, mark=*, only marks, mark size = 3]
table{%
x  y
0 84636.6251771964
1 149122.976813395
2 154459.385303937
3 191546.809775256
4 192254.411176995
5 195582.069008292
6 189563.040948544
7 182411.350641375
8 193031.038009356
9 201709.080147046
};
\addplot [line width=1.08pt, color0, forget plot]
table {%
0 84636.6251771964
1 149122.976813395
2 154459.385303937
3 191546.809775256
4 192254.411176995
5 195582.069008292
6 189563.040948544
7 182411.350641375
8 193031.038009356
9 201709.080147046
};
\addplot [line width=1.08pt, color0, forget plot]
table {%
0 77886.0361965227
0 91246.4746717489
};
\addplot [line width=1.08pt, color0, forget plot]
table {%
1 133820.48880585
1 169430.446465029
};
\addplot [line width=1.08pt, color0, forget plot]
table {%
2 139388.152966167
2 171762.548982563
};
\addplot [line width=1.08pt, color0, forget plot]
table {%
3 164014.492555798
3 229212.658166506
};
\addplot [line width=1.08pt, color0, forget plot]
table {%
4 162734.005396178
4 236075.414585923
};
\addplot [line width=1.08pt, color0, forget plot]
table {%
5 178374.325556508
5 214118.378219963
};
\addplot [line width=1.08pt, color0, forget plot]
table {%
6 174420.157740441
6 205093.854193899
};
\addplot [line width=1.08pt, color0, forget plot]
table {%
7 169721.133250983
7 196940.488462631
};
\addplot [line width=1.08pt, color0, forget plot]
table {%
8 177266.033664892
8 210006.779072766
};
\addplot [line width=1.08pt, color0, forget plot]
table {%
9 169561.907872108
9 257213.630842998
};
\addplot [draw=color1, fill=color1, mark=square, only marks, mark size = 3]
table{%
x  y
0 152223.841994641
1 236925.270264474
2 284979.726873317
3 302830.080860407
4 341483.066329842
5 353147.274920097
6 327839.173087077
7 345130.442526884
8 337357.997917831
9 348054.124125912
};
\addplot [line width=1.08pt, color1, forget plot]
table {%
0 152223.841994641
1 236925.270264474
2 284979.726873317
3 302830.080860407
4 341483.066329842
5 353147.274920097
6 327839.173087077
7 345130.442526884
8 337357.997917831
9 348054.124125912
};
\addplot [line width=1.08pt, color1, forget plot]
table {%
0 140647.503673646
0 164758.893749261
};
\addplot [line width=1.08pt, color1, forget plot]
table {%
1 220358.82157815
1 252867.50095676
};
\addplot [line width=1.08pt, color1, forget plot]
table {%
2 262346.829989489
2 309168.370890578
};
\addplot [line width=1.08pt, color1, forget plot]
table {%
3 276132.077238423
3 330980.516379929
};
\addplot [line width=1.08pt, color1, forget plot]
table {%
4 308222.504414981
4 380688.988055758
};
\addplot [line width=1.08pt, color1, forget plot]
table {%
5 315730.592102488
5 402591.378418507
};
\addplot [line width=1.08pt, color1, forget plot]
table {%
6 304297.58514414
6 354376.57553583
};
\addplot [line width=1.08pt, color1, forget plot]
table {%
7 320728.374650279
7 372536.935151671
};
\addplot [line width=1.08pt, color1, forget plot]
table {%
8 312852.339625267
8 361980.92628874
};
\addplot [line width=1.08pt, color1, forget plot]
table {%
9 326250.074627571
9 369483.039822774
};
\addplot [draw=color2, fill=color2, mark=o, only marks, mark size = 3]
table{%
x  y
0 172579.014669415
1 273078.14745203
2 331689.355084594
3 366445.963383035
4 369598.202567219
5 391269.42439877
6 404594.338287295
7 387693.38524891
8 455252.420987747
9 385543.318740716
};
\addplot [line width=1.08pt, color2, forget plot]
table {%
0 172579.014669415
1 273078.14745203
2 331689.355084594
3 366445.963383035
4 369598.202567219
5 391269.42439877
6 404594.338287295
7 387693.38524891
8 455252.420987747
9 385543.318740716
};
\addplot [line width=1.08pt, color2, forget plot]
table {%
0 159354.407044962
0 187509.179190888
};
\addplot [line width=1.08pt, color2, forget plot]
table {%
1 254046.90387289
1 293159.34079222
};
\addplot [line width=1.08pt, color2, forget plot]
table {%
2 305069.637554098
2 362212.036042454
};
\addplot [line width=1.08pt, color2, forget plot]
table {%
3 331550.851982407
3 405422.808803531
};
\addplot [line width=1.08pt, color2, forget plot]
table {%
4 347103.781841548
4 394788.727415941
};
\addplot [line width=1.08pt, color2, forget plot]
table {%
5 351972.384910706
5 435675.873917753
};
\addplot [line width=1.08pt, color2, forget plot]
table {%
6 369607.735146538
6 443551.210257207
};
\addplot [line width=1.08pt, color2, forget plot]
table {%
7 350353.315801426
7 431436.553732054
};
\addplot [line width=1.08pt, color2, forget plot]
table {%
8 399858.788588162
8 533985.609288723
};
\addplot [line width=1.08pt, color2, forget plot]
table {%
9 356227.952654443
9 416454.293956437
};

\nextgroupplot[
axis line style={white!80!black},
legend cell align={left},
legend style={fill opacity=0.8, draw opacity=1, text opacity=1, draw=white!80!black, font=\large},
log basis y={10},
tick align=inside,
x grid style={white!80!black},
xlabel={\large Number of nodes \(\displaystyle n\)},
xmajorgrids,
xmajorticks=true,
xmin=-0.5, xmax=9.5,
xminorgrids,
xtick style={color=white!15!black},
xtick={0,1,2,3,4,5,6,7,8,9},
xticklabels={200,400,600,800,1000,1200,1400,1600,1800,2000},
xticklabel style={rotate=90},
xtick pos=left,
y grid style={white!80!black},
ymajorgrids,
ymajorticks=true,
ymin=0.005, ymax=0.04,
yminorgrids,
ytick style={color=white!15!black},
yticklabel style={font=\large},
ytick pos=left,
width = 8cm,
height = 6cm
]
\addplot [draw=color0, fill=color0, mark=*, only marks, mark size = 3]
table{%
x  y
0 0.0111555600266746
1 0.00994529582268556
2 0.00930671386729618
3 0.00902946206972036
4 0.00876675106439742
5 0.00859598687183526
6 0.00842157360105946
7 0.00842026056850737
8 0.00819458489412961
9 0.00810976324875136
};
\addlegendentry{$p_{\sf re}$ = 0.1}
\addplot [line width=1.08pt, color0, forget plot]
table {%
0 0.0111555600266746
1 0.00994529582268556
2 0.00930671386729618
3 0.00902946206972036
4 0.00876675106439742
5 0.00859598687183526
6 0.00842157360105946
7 0.00842026056850737
8 0.00819458489412961
9 0.00810976324875136
};
\addplot [line width=1.08pt, color0, forget plot]
table {%
0 0.0107922009222172
0 0.0115336632967854
};
\addplot [line width=1.08pt, color0, forget plot]
table {%
1 0.00975778826887879
1 0.0101459057567516
};
\addplot [line width=1.08pt, color0, forget plot]
table {%
2 0.00912670386220252
2 0.00948295625926121
};
\addplot [line width=1.08pt, color0, forget plot]
table {%
3 0.00889835167695728
3 0.00916756178059376
};
\addplot [line width=1.08pt, color0, forget plot]
table {%
4 0.00865080675744907
4 0.0088972863403467
};
\addplot [line width=1.08pt, color0, forget plot]
table {%
5 0.00849252750822559
5 0.00871871180533893
};
\addplot [line width=1.08pt, color0, forget plot]
table {%
6 0.00831128461655299
6 0.00854380666111256
};
\addplot [line width=1.08pt, color0, forget plot]
table {%
7 0.00831813374038733
7 0.00853783816972761
};
\addplot [line width=1.08pt, color0, forget plot]
table {%
8 0.00811180146783414
8 0.00827969139220988
};
\addplot [line width=1.08pt, color0, forget plot]
table {%
9 0.00802900288740386
9 0.00819383154736759
};
\addplot [draw=color1, fill=color1, mark=square, only marks, mark size = 3]
table{%
x  y
0 0.024800844354858
1 0.022370964235506
2 0.021007805721124
3 0.0204904681204029
4 0.0196850878581161
5 0.0192571050226535
6 0.01882118720747
7 0.018549108753888
8 0.0180615123488957
9 0.0179988917583812
};
\addlegendentry{$p_{\sf re}$ = 0.5}
\addplot [line width=1.08pt, color1, forget plot]
table {%
0 0.024800844354858
1 0.022370964235506
2 0.021007805721124
3 0.0204904681204029
4 0.0196850878581161
5 0.0192571050226535
6 0.01882118720747
7 0.018549108753888
8 0.0180615123488957
9 0.0179988917583812
};
\addplot [line width=1.08pt, color1, forget plot]
table {%
0 0.0238610788740566
0 0.0258002551773358
};
\addplot [line width=1.08pt, color1, forget plot]
table {%
1 0.0218551868666719
1 0.0229075288255567
};
\addplot [line width=1.08pt, color1, forget plot]
table {%
2 0.0205398615476623
2 0.021456342500142
};
\addplot [line width=1.08pt, color1, forget plot]
table {%
3 0.0200895870071103
3 0.020922127356991
};
\addplot [line width=1.08pt, color1, forget plot]
table {%
4 0.0193619467960632
4 0.0200286542453595
};
\addplot [line width=1.08pt, color1, forget plot]
table {%
5 0.0190065136240656
5 0.0195113699315279
};
\addplot [line width=1.08pt, color1, forget plot]
table {%
6 0.0186092124296869
6 0.0190419158477072
};
\addplot [line width=1.08pt, color1, forget plot]
table {%
7 0.018320688354679
7 0.0188096129295858
};
\addplot [line width=1.08pt, color1, forget plot]
table {%
8 0.0178824268556351
8 0.0182518884898521
};
\addplot [line width=1.08pt, color1, forget plot]
table {%
9 0.0178156621617274
9 0.0181825960200303
};
\addplot [draw=color2, fill=color2, mark=o, only marks, mark size = 3]
table{%
x  y
0 0.0354038933680986
1 0.029408693057821
2 0.0277538155081099
3 0.0272343922162342
4 0.0263389285685485
5 0.0260445347684289
6 0.0246694832092637
7 0.0244355998700362
8 0.0244557352715339
9 0.0243222509558819
};
\addlegendentry{$p_{\sf re}$ = 0.9}
\addplot [line width=1.08pt, color2, forget plot]
table {%
0 0.0354038933680986
1 0.029408693057821
2 0.0277538155081099
3 0.0272343922162342
4 0.0263389285685485
5 0.0260445347684289
6 0.0246694832092637
7 0.0244355998700362
8 0.0244557352715339
9 0.0243222509558819
};
\addplot [line width=1.08pt, color2, forget plot]
table {%
0 0.0339537725959263
0 0.0370062655093996
};
\addplot [line width=1.08pt, color2, forget plot]
table {%
1 0.0286108592447194
1 0.0302702136420392
};
\addplot [line width=1.08pt, color2, forget plot]
table {%
2 0.0271537838283673
2 0.0283927205116471
};
\addplot [line width=1.08pt, color2, forget plot]
table {%
3 0.0267093213152462
3 0.0277615875866199
};
\addplot [line width=1.08pt, color2, forget plot]
table {%
4 0.0259622333661865
4 0.0267268093102903
};
\addplot [line width=1.08pt, color2, forget plot]
table {%
5 0.0256987151592046
5 0.0264723015871522
};
\addplot [line width=1.08pt, color2, forget plot]
table {%
6 0.0243610169148455
6 0.0249800592539279
};
\addplot [line width=1.08pt, color2, forget plot]
table {%
7 0.0241590245094596
7 0.0247249531529757
};
\addplot [line width=1.08pt, color2, forget plot]
table {%
8 0.0241521789457391
8 0.0247648056656863
};
\addplot [line width=1.08pt, color2, forget plot]
table {%
9 0.0240598341927692
9 0.0245921030438158
};
\end{groupplot}

\end{tikzpicture}

%% file: fig_D_n_2.tex
\begin{tikzpicture}

\definecolor{color0}{rgb}{0.12156862745098,0.466666666666667,0.705882352941177}
\definecolor{color1}{rgb}{1,0.498039215686275,0.0549019607843137}
\definecolor{color2}{rgb}{0.172549019607843,0.627450980392157,0.172549019607843}

\begin{groupplot}[group style={group size=2 by 1}]
\nextgroupplot[
axis line style={white!80!black},
legend cell align={left},
legend style={
  fill opacity=0.8,
  draw opacity=1,
  text opacity=1,
  at={(0.91,0.5)},
  anchor=east,
  draw=white!80!black,
  font = \large
},
log basis y={10},
tick align=inside,
x grid style={white!80!black},
xlabel={\large Number of nodes \(\displaystyle n\)},
xmajorgrids,
xmajorticks=true,
xmin=-0.5, xmax=9.5,
xminorgrids,
xtick style={color=white!15!black},
xtick={0,1,2,3,4,5,6,7,8,9},
xticklabels={200,400,600,800,1000,1200,1400,1600,1800,2000},
xticklabel style={rotate=90},
xtick pos=left,
y grid style={white!80!black},
ylabel={\large \(\displaystyle \mathbb{D}_{\cal H} ( {\bf L}_{\sf norm}, \hat{\bf L}_{\sf norm} )\)},
ymajorgrids,
ymajorticks=true,
ymin=0.2, ymax=1,
yminorgrids,
ytick style={color=white!15!black},
yticklabel style={font=\large},
ytick pos=left,
width = 8cm,
height = 5.5cm
]
\addplot [draw=color0, fill=color0, mark=*, only marks, mark size = 3]
table{%
x  y
0 0.381230639701382
1 0.362267194927996
2 0.350952339777826
3 0.344052896770379
4 0.338275107497599
5 0.334258056896496
6 0.330074969130927
7 0.326595477460074
8 0.324342064011572
9 0.32217696576926
};
\addplot [line width=1.08pt, color0, forget plot]
table {%
0 0.381230639701382
1 0.362267194927996
2 0.350952339777826
3 0.344052896770379
4 0.338275107497599
5 0.334258056896496
6 0.330074969130927
7 0.326595477460074
8 0.324342064011572
9 0.32217696576926
};
\addplot [line width=1.08pt, color0, forget plot]
table {%
0 0.379292185616256
0 0.383352800188159
};
\addplot [line width=1.08pt, color0, forget plot]
table {%
1 0.360908059558124
1 0.363669090868135
};
\addplot [line width=1.08pt, color0, forget plot]
table {%
2 0.349924414570421
2 0.352100641304283
};
\addplot [line width=1.08pt, color0, forget plot]
table {%
3 0.343178930238805
3 0.345091316359753
};
\addplot [line width=1.08pt, color0, forget plot]
table {%
4 0.337547647431659
4 0.338954783490858
};
\addplot [line width=1.08pt, color0, forget plot]
table {%
5 0.333603816047917
5 0.334987169338839
};
\addplot [line width=1.08pt, color0, forget plot]
table {%
6 0.329471872591735
6 0.33069308223311
};
\addplot [line width=1.08pt, color0, forget plot]
table {%
7 0.326031935626577
7 0.327244600890732
};
\addplot [line width=1.08pt, color0, forget plot]
table {%
8 0.323819940422682
8 0.324885914581746
};
\addplot [line width=1.08pt, color0, forget plot]
table {%
9 0.32171486112517
9 0.322625703481709
};
\addplot [draw=color1, fill=color1, mark=square, only marks, mark size = 3]
table{%
x  y
0 0.736854136483426
1 0.73777195878317
2 0.728113717067633
3 0.719213312961915
4 0.710033281604537
5 0.703645719596839
6 0.697496831618279
7 0.692365701508665
8 0.68791614327562
9 0.68387862985237
};
\addplot [line width=1.08pt, color1, forget plot]
table {%
0 0.736854136483426
1 0.73777195878317
2 0.728113717067633
3 0.719213312961915
4 0.710033281604537
5 0.703645719596839
6 0.697496831618279
7 0.692365701508665
8 0.68791614327562
9 0.68387862985237
};
\addplot [line width=1.08pt, color1, forget plot]
table {%
0 0.733914015034589
0 0.739956206060624
};
\addplot [line width=1.08pt, color1, forget plot]
table {%
1 0.735931430012114
1 0.739491295635901
};
\addplot [line width=1.08pt, color1, forget plot]
table {%
2 0.726668331389251
2 0.729600435083414
};
\addplot [line width=1.08pt, color1, forget plot]
table {%
3 0.717997711884906
3 0.720442946504643
};
\addplot [line width=1.08pt, color1, forget plot]
table {%
4 0.709029006389329
4 0.710975793174737
};
\addplot [line width=1.08pt, color1, forget plot]
table {%
5 0.702843869784544
5 0.704378980296648
};
\addplot [line width=1.08pt, color1, forget plot]
table {%
6 0.696823938927289
6 0.698137605307047
};
\addplot [line width=1.08pt, color1, forget plot]
table {%
7 0.691708115719089
7 0.692997051932624
};
\addplot [line width=1.08pt, color1, forget plot]
table {%
8 0.687290927006335
8 0.688585217305316
};
\addplot [line width=1.08pt, color1, forget plot]
table {%
9 0.683332243533033
9 0.684424177322142
};
\addplot [draw=color2, fill=color2, mark=o, only marks, mark size = 3]
table{%
x  y
0 0.9071731848072
1 0.940107089302127
2 0.937939817633872
3 0.932113431178076
4 0.925342690993465
5 0.919478732965315
6 0.913718513575583
7 0.907634222650699
8 0.90363046431944
9 0.899449774850544
};
\addplot [line width=1.08pt, color2, forget plot]
table {%
0 0.9071731848072
1 0.940107089302127
2 0.937939817633872
3 0.932113431178076
4 0.925342690993465
5 0.919478732965315
6 0.913718513575583
7 0.907634222650699
8 0.90363046431944
9 0.899449774850544
};
\addplot [line width=1.08pt, color2, forget plot]
table {%
0 0.90409246863906
0 0.910505613410005
};
\addplot [line width=1.08pt, color2, forget plot]
table {%
1 0.938209498736958
1 0.942046174527896
};
\addplot [line width=1.08pt, color2, forget plot]
table {%
2 0.936390675284899
2 0.939703753212366
};
\addplot [line width=1.08pt, color2, forget plot]
table {%
3 0.930654469506775
3 0.933568193489977
};
\addplot [line width=1.08pt, color2, forget plot]
table {%
4 0.92429799335992
4 0.92642361531571
};
\addplot [line width=1.08pt, color2, forget plot]
table {%
5 0.918450560926173
5 0.920551544268068
};
\addplot [line width=1.08pt, color2, forget plot]
table {%
6 0.912875591916363
6 0.91461140658707
};
\addplot [line width=1.08pt, color2, forget plot]
table {%
7 0.906867466354058
7 0.908396651450716
};
\addplot [line width=1.08pt, color2, forget plot]
table {%
8 0.902922753130456
8 0.90433799391002
};
\addplot [line width=1.08pt, color2, forget plot]
table {%
9 0.898746862433734
9 0.900243041451971
};

\nextgroupplot[
axis line style={white!80!black},
legend cell align={left},
legend style={
  fill opacity=0.8,
  draw opacity=1,
  text opacity=1,
  at={(0.97,0.85)},
  anchor=east,
  draw=white!80!black,
  font = \large,
  legend columns = 3,
},
log basis y={10},
tick align=inside,
x grid style={white!80!black},
xlabel={\large Number of nodes \(\displaystyle n\)},
xmajorgrids,
xmajorticks=true,
xmin=-0.5, xmax=9.5,
xminorgrids,
xtick style={color=white!15!black},
xtick={0,1,2,3,4,5,6,7,8,9},
xticklabels={200,400,600,800,1000,1200,1400,1600,1800,2000},
xticklabel style={rotate=90},
y grid style={white!80!black},
xtick pos=left,
ymajorgrids,
ymajorticks=true,
ymin=0.0, ymax=0.15,
yminorgrids,
ytick style={color=white!15!black},
ytick={0,0.05, 0.1, 0.15},
yticklabels={$0$,$0.05$, $0.1$, $0.15$},
yticklabel style={font=\large},
ytick pos=left,
width = 8cm,
height = 5.5cm
]
\addplot [draw=color0, fill=color0, mark=*, only marks, mark size = 3]
table{%
x  y
0 0.0508411943005702
1 0.048655802444197
2 0.0473069934130519
3 0.0463931690411028
4 0.0456569799538332
5 0.0450960347420985
6 0.0445891666680406
7 0.044108371053643
8 0.0438650496554766
9 0.0435136095749908
};
\addplot [line width=1.08pt, color0, forget plot]
table {%
0 0.0508411943005702
1 0.048655802444197
2 0.0473069934130519
3 0.0463931690411028
4 0.0456569799538332
5 0.0450960347420985
6 0.0445891666680406
7 0.044108371053643
8 0.0438650496554766
9 0.0435136095749908
};
\addplot [line width=1.08pt, color0, forget plot]
table {%
0 0.0505633045385904
0 0.0511172513017237
};
\addplot [line width=1.08pt, color0, forget plot]
table {%
1 0.0484541835517146
1 0.0488659644378959
};
\addplot [line width=1.08pt, color0, forget plot]
table {%
2 0.0471852622809762
2 0.0474521689939004
};
\addplot [line width=1.08pt, color0, forget plot]
table {%
3 0.0462563840962458
3 0.046533776421628
};
\addplot [line width=1.08pt, color0, forget plot]
table {%
4 0.0455586477115832
4 0.0457562515365667
};
\addplot [line width=1.08pt, color0, forget plot]
table {%
5 0.045002389043276
5 0.0451823663558051
};
\addplot [line width=1.08pt, color0, forget plot]
table {%
6 0.0445111118133815
6 0.0446688413591768
};
\addplot [line width=1.08pt, color0, forget plot]
table {%
7 0.0440363761632068
7 0.0441842070475183
};
\addplot [line width=1.08pt, color0, forget plot]
table {%
8 0.0438021439448931
8 0.043931969854664
};
\addplot [line width=1.08pt, color0, forget plot]
table {%
9 0.0434482594036406
9 0.0435754431146531
};
\addplot [draw=color1, fill=color1, mark=square, only marks, mark size = 3]
table{%
x  y
0 0.0986261281318187
1 0.0992268515933866
2 0.0981721906067296
3 0.096890760387799
4 0.0959577008642199
5 0.0950663313866873
6 0.0943074592227413
7 0.0935800483906928
8 0.0930104564481738
9 0.0925430106507906
};
\addplot [line width=1.08pt, color1, forget plot]
table {%
0 0.0986261281318187
1 0.0992268515933866
2 0.0981721906067296
3 0.096890760387799
4 0.0959577008642199
5 0.0950663313866873
6 0.0943074592227413
7 0.0935800483906928
8 0.0930104564481738
9 0.0925430106507906
};
\addplot [line width=1.08pt, color1, forget plot]
table {%
0 0.0982767875166145
0 0.0990300986013737
};
\addplot [line width=1.08pt, color1, forget plot]
table {%
1 0.0989752222207185
1 0.0994753316950035
};
\addplot [line width=1.08pt, color1, forget plot]
table {%
2 0.0980003183865085
2 0.0983543439388458
};
\addplot [line width=1.08pt, color1, forget plot]
table {%
3 0.09674208119436
3 0.0970494821928171
};
\addplot [line width=1.08pt, color1, forget plot]
table {%
4 0.0958280704666094
4 0.0960901739692553
};
\addplot [line width=1.08pt, color1, forget plot]
table {%
5 0.0949655516426944
5 0.0951606275515064
};
\addplot [line width=1.08pt, color1, forget plot]
table {%
6 0.0942120014356729
6 0.094407856990937
};
\addplot [line width=1.08pt, color1, forget plot]
table {%
7 0.0934917758395638
7 0.0936725815835491
};
\addplot [line width=1.08pt, color1, forget plot]
table {%
8 0.0929126787466579
8 0.0931001598167438
};
\addplot [line width=1.08pt, color1, forget plot]
table {%
9 0.0924556111101799
9 0.092623435531646
};
\addplot [draw=color2, fill=color2, mark=o, only marks, mark size = 3]
table{%
x  y
0 0.121639340001556
1 0.126380662565269
2 0.126394442562423
3 0.12573161724636
4 0.124789453711719
5 0.124117951657479
6 0.123435031336509
7 0.122661569614198
8 0.122114609328075
9 0.121562216838133
};
\addplot [line width=1.08pt, color2, forget plot]
table {%
0 0.121639340001556
1 0.126380662565269
2 0.126394442562423
3 0.12573161724636
4 0.124789453711719
5 0.124117951657479
6 0.123435031336509
7 0.122661569614198
8 0.122114609328075
9 0.121562216838133
};
\addplot [line width=1.08pt, color2, forget plot]
table {%
0 0.121171878522753
0 0.122123275511634
};
\addplot [line width=1.08pt, color2, forget plot]
table {%
1 0.12607558283941
1 0.12667995681776
};
\addplot [line width=1.08pt, color2, forget plot]
table {%
2 0.126183077861944
2 0.126607683334328
};
\addplot [line width=1.08pt, color2, forget plot]
table {%
3 0.125527047890495
3 0.125948760646948
};
\addplot [line width=1.08pt, color2, forget plot]
table {%
4 0.124621066355989
4 0.124957214643547
};
\addplot [line width=1.08pt, color2, forget plot]
table {%
5 0.123984242641677
5 0.124249360104928
};
\addplot [line width=1.08pt, color2, forget plot]
table {%
6 0.123303037091328
6 0.123557559369475
};
\addplot [line width=1.08pt, color2, forget plot]
table {%
7 0.122550706553238
7 0.122770135544507
};
\addplot [line width=1.08pt, color2, forget plot]
table {%
8 0.122014082938504
8 0.122212587086092
};
\addplot [line width=1.08pt, color2, forget plot]
table {%
9 0.1214783605811
9 0.121665083179413
};
\end{groupplot}

\end{tikzpicture}

%% file: fig_D_p_strong.tex
\begin{tikzpicture}

\definecolor{color0}{rgb}{0.12156862745098,0.466666666666667,0.705882352941177}
\definecolor{color1}{rgb}{1,0.498039215686275,0.0549019607843137}
\definecolor{color2}{rgb}{0.172549019607843,0.627450980392157,0.172549019607843}
\definecolor{color3}{rgb}{0.83921568627451,0.152941176470588,0.156862745098039}

\begin{axis}[
axis line style={white!80!black},
legend cell align={left},
legend style={
  fill opacity=0.8,
  draw opacity=1,
  text opacity=1,
  at={(0.9,0.475)},
  anchor=east,
  draw=white!80!black,
  font = \tiny
},
log basis y={10},
tick align=inside,
x grid style={white!80!black},
xlabel={\footnotesize Rewiring ratio \(\displaystyle p_{\sf re}\)},
xmajorgrids,
xmajorticks=true,
xmin=-0.5, xmax=4.5,
xminorgrids,
xtick style={color=white!15!black},
xtick={0,1,2,3,4},
xticklabels={0.01,0.05,0.1,0.15,0.2},
xtick pos=left,
y grid style={white!80!black},
ylabel={\footnotesize \(\displaystyle \mathbb{D}_{\cal H} ( {\bf L}_U, \hat{\bf L}_U )\)},
ymajorgrids,
ymajorticks=true,
yminorgrids,
ymode=log,
ytick style={color=white!15!black},
ytick pos=left,
width=4.5cm,
height=4.5cm
]
\addplot [draw=color2, fill=color2, mark=*, only marks]
table{%
x  y
0 6.65095850922261e+20
1 3.99209829338645e+21
2 5.24009177602074e+21
3 5.46976260985139e+21
4 5.51991812621937e+21
};
\addlegendentry{$\exp( {\bf L}_U / \log n )$}
\addplot [line width=1.08pt, color2, forget plot]
table {%
0 6.65095850922261e+20
1 3.99209829338645e+21
2 5.24009177602074e+21
3 5.46976260985139e+21
4 5.51991812621937e+21
};
\addplot [line width=1.08pt, color2, forget plot]
table {%
0 5.57276264313396e+20
0 7.69369514077068e+20
};
\addplot [line width=1.08pt, color2, forget plot]
table {%
1 3.80362208721171e+21
1 4.15122998761899e+21
};
\addplot [line width=1.08pt, color2, forget plot]
table {%
2 5.19483834183389e+21
2 5.28256924148562e+21
};
\addplot [line width=1.08pt, color2, forget plot]
table {%
3 5.45680016263916e+21
3 5.48173701325721e+21
};
\addplot [line width=1.08pt, color2, forget plot]
table {%
4 5.51809407857401e+21
4 5.52148351012084e+21
};
\addplot [draw=color3, fill=color3, mark=*, only marks]
table{%
x  y
0 0.241721074328827
1 0.474258338918389
2 0.657452155513072
3 0.778473977028639
4 0.852789928283214
};
\addlegendentry{$\exp(-{\bf L}_U / \log n)$}
\addplot [line width=1.08pt, color3, forget plot]
table {%
0 0.241721074328827
1 0.474258338918389
2 0.657452155513072
3 0.778473977028639
4 0.852789928283214
};
\addplot [line width=1.08pt, color3, forget plot]
table {%
0 0.234890032386545
0 0.248958225624746
};
\addplot [line width=1.08pt, color3, forget plot]
table {%
1 0.464762498991842
1 0.483690059800956
};
\addplot [line width=1.08pt, color3, forget plot]
table {%
2 0.649284906607309
2 0.667483400486463
};
\addplot [line width=1.08pt, color3, forget plot]
table {%
3 0.770435617728909
3 0.786916296042255
};
\addplot [line width=1.08pt, color3, forget plot]
table {%
4 0.844309394018802
4 0.861863625339757
};
\end{axis}

\end{tikzpicture}

%% file: fig_D_p_weak.tex
\begin{tikzpicture}

\definecolor{color0}{rgb}{0.12156862745098,0.466666666666667,0.705882352941177}
\definecolor{color1}{rgb}{1,0.498039215686275,0.0549019607843137}
\definecolor{color2}{rgb}{0.172549019607843,0.627450980392157,0.172549019607843}
\definecolor{color3}{rgb}{0.83921568627451,0.152941176470588,0.156862745098039}

\begin{axis}[
axis line style={white!80!black},
legend cell align={left},
legend style={
  fill opacity=0.8,
  draw opacity=1,
  text opacity=1,
  at={(0.975,0.475)},
  anchor=east,
  draw=white!80!black,
  font = \tiny
},
log basis y={10},
tick align=inside,
x grid style={white!80!black},
xlabel={\footnotesize Rewiring ratio \(\displaystyle p_{\sf re}\)},
xmajorgrids,
xmajorticks=true,
xmin=-0.5, xmax=4.5,
xminorgrids,
xtick style={color=white!15!black},
xtick={0,1,2,3,4},
xticklabels={0.01,0.05,0.1,0.15,0.2},
xtick pos=left,
y grid style={white!80!black},
ymajorgrids,
ymajorticks=true,
yminorgrids,
ymode=log,
ytick style={color=white!15!black},
ytick pos=left,
width=4.5cm,
height=4.5cm
]
\addplot [draw=color2, fill=color2, mark=*, only marks]
table{%
x  y
0 2.11266446865847
1 20.7267131948314
2 41.6284185757132
3 59.3987295661545
4 74.4902245785305
};
\addlegendentry{$\exp(0.1 {\bf L}_U / \log n )$}
\addplot [line width=1.08pt, color2, forget plot]
table {%
0 2.11266446865847
1 20.7267131948314
2 41.6284185757132
3 59.3987295661545
4 74.4902245785305
};
\addplot [line width=1.08pt, color2, forget plot]
table {%
0 1.80933031511192
0 2.4210596109125
};
\addplot [line width=1.08pt, color2, forget plot]
table {%
1 19.3037065753549
1 22.0661933213993
};
\addplot [line width=1.08pt, color2, forget plot]
table {%
2 39.8861757359869
2 43.4605313198489
};
\addplot [line width=1.08pt, color2, forget plot]
table {%
3 57.3755409219987
3 61.4210024712754
};
\addplot [line width=1.08pt, color2, forget plot]
table {%
4 72.8322237953422
4 76.0460186082606
};
\addplot [draw=color3, fill=color3, mark=*, only marks]
table{%
x  y
0 0.0412759157143382
1 0.107493729703908
2 0.196538128460606
3 0.266438832887954
4 0.337550695394556
};
\addlegendentry{$\exp(-0.1 {\bf L}_U / \log n)$}
\addplot [line width=1.08pt, color3, forget plot]
table {%
0 0.0412759157143382
1 0.107493729703908
2 0.196538128460606
3 0.266438832887954
4 0.337550695394556
};
\addplot [line width=1.08pt, color3, forget plot]
table {%
0 0.0401912489037516
0 0.042405073731613
};
\addplot [line width=1.08pt, color3, forget plot]
table {%
1 0.104248720662779
1 0.11085546187772
};
\addplot [line width=1.08pt, color3, forget plot]
table {%
2 0.191701250716831
2 0.201591311904228
};
\addplot [line width=1.08pt, color3, forget plot]
table {%
3 0.260867089491712
3 0.272713165448456
};
\addplot [line width=1.08pt, color3, forget plot]
table {%
4 0.331937658644394
4 0.343454759756389
};
\end{axis}

\end{tikzpicture}